\documentclass{IEEEtran}
\usepackage{amsmath}
\usepackage[utf8]{inputenc} 
\usepackage{epsfig,scalefnt,multirow}
\usepackage{url}
\usepackage{ifthen}
\usepackage{mathtools}
\usepackage{cite}
\usepackage{graphicx}
\usepackage{amssymb}
\usepackage{tabularx}
\usepackage{amsmath}
\usepackage{epstopdf}
\usepackage{epsf}
\usepackage{algorithm}
\usepackage{algpseudocode}
\usepackage{algpascal}
\usepackage{cases}
\usepackage{caption}
\usepackage[labelformat=simple]{subcaption}

\usepackage{stfloats}
\usepackage{float}
\usepackage{xcolor}
\usepackage{tabularx}

\usepackage{epsfig,amsmath,amssymb,epsf,amsthm,scalefnt,multirow}
\usepackage{xcolor}
\usepackage{float}
\usepackage{cite}
\usepackage{psfrag}
\usepackage{gensymb}
\usepackage{cleveref}
\usepackage{mathrsfs}
\usepackage{mathtools}
\usepackage{algpseudocode}
\usepackage{algorithm}
\usepackage{enumerate}

  \def\cC{{\mathcal{C}}} 
   
  \def\cK{{\mathcal{K}}} 
 \def\cN{{\mathcal{N}}} \def\cO{{\mathcal{O}}} 
   
 \def\cV{{\mathcal{V}}}

\def\trace{\mathop{\mathrm{tr}}}

\def\Re{\mathop{\mathrm{Re}}}

\def\bSigma{{\pmb{\Sigma}}} 
 
 \def\bgamma{{\pmb{\gamma}}}

\def\b0{{\pmb{0}}} 

\def\ba{{\mathbf{a}}}   
 \def\bff{{\mathbf{f}}}  \def\bh{{\mathbf{h}}}
   
\def\bm{{\mathbf{m}}}   
  \def\bs{{\mathbf{s}}} 
\def\bu{{\mathbf{u}}}   \def\bx{{\mathbf{x}}}

\def\bA{{\mathbf{A}}}   
 \def\bF{{\mathbf{F}}}  \def\bH{{\mathbf{H}}}
\def\bI{{\mathbf{I}}}   
\def\bM{{\mathbf{M}}}   
 \def\bR{{\mathbf{R}}}

\DeclarePairedDelimiter\norm{\lVert}{\rVert}

\newtheorem*{lemma*}{Lemma}

\newtheorem{proposition}{Proposition}

\begin{document}
	%

	\title{Smart Resource Allocation at mmWave/THz Frequencies with Cooperative Rate-Splitting}

	%
	%
	%
	
			\author{Hyesang Cho and~Junil Choi
		\thanks{		
		This work was supported by Institute of Information \& Communications Technology Planning \& Evaluation (IITP) grant funded by the Korea government (MSIT) (No.2021-000269, Development of sub-THz band wireless transmission and access core technology for 6G Tbps data rate) and the MSIT under ITRC (Information Technology Research Center) support program (IITP-2024-2020-0-01787) supervised by IITP.	
		
		H. Cho and J. Choi are with the School of Electrical Engineering, Korea Advanced Institute of Science and Technology, Daejeon 34141, South Korea (e-mail: \{nanjohn96, junil\}@kaist.ac.kr).
		}
	}

	\maketitle
	

	\begin{abstract} \label{sec:abs} 
		In this paper, we propose algorithms to minimize the energy consumption in millimeter wave/terahertz multi-user downlink communication systems. 
		To ensure coverage in blockage-vulnerable high frequency systems, we consider cooperative rate-splitting (CRS) and transmission over multiple time blocks, where via CRS, multiple users cooperate to assist a blocked user.
		Moreover, we show that transmission over multiple time blocks provides benefits through smart resource allocation.
		We first propose a communication framework named improved distinct extraction-based CRS (iDeCRS) that utilizes the benefits of rate-splitting.
		With our transmission framework, we derive a performance benchmark assuming genie channel state information (CSI), i.e., the channels of the present and future time blocks are known, denoted as GENIE.
		Using the results from GENIE, we derive a novel efficiency constrained optimization (ECO) algorithm assuming instantaneous CSI.
		In addition, a simple but effective even data transmission (EDT) algorithm that promotes steady transmission along the time blocks is proposed.
		Simulation results show that ECO and EDT have satisfactory performances compared to GENIE.
		The results also show that ECO outperforms EDT when many users are cooperating, and vise versa.
	\end{abstract}
	
	
	


	
	\section{Introduction}\label{sec:intro}
	Increasing demands of high data rates for mobile devices have triggered a variety of new fields in wireless communication such as massive multiple-input multiple-output (MIMO), reconfigurable intelligent surfaces (RISs), and cell-free communication \cite{mmWave1, Tech1}.
	Among emerging technologies, a simple but powerful approach is to increase the carrier frequency into millimeter wave (mmWave) and terahertz (THz) frequencies \cite{mmWave2,THz1,THz2}.
	By using the broad and vacant frequency bands, mmWave and THz communication systems can achieve extremely high data rates, where a simple THz communication system has achieved a $100$ Gbps data rate in \cite{THzbad3}.
	
	Although with their strengths, high frequency systems have their downsides. 
	Due to the high frequency signals, the multi-path effect decreases substantially, leading to line-of-sight (LoS) dominant channels \cite{Chaccour2022}. 
	Hence, if the LoS link between an access point (AP) and a user equipment (UE) is blocked, serving the UE directly is not possible.
	The high frequency signals also experience extreme propagation loss, and the hardware specifications indicate that the transmitters for THz frequency signals will have low transmit power \cite{THzbad3, THzbad4, THzbad5}.
	Considering the features of high frequency systems, we can expect that restricted coverage will be a major challenge.
	
	To overcome this issue, we consider cooperative communication. 
	Cooperative communication is a concept where the UEs cooperate with each other to improve the overall performance of the system, e.g., the UEs relay or share information \cite{Coop1, Coop2}.
	Different from relaying systems, cooperative communication does not use dedicated relays but uses the UEs for cooperation.
	Therefore, the cooperating UEs not only relay information for other UEs, but receive information for themselves.
	While the cooperating UEs must use their resources for other UEs, the overall performance of the system can improve through cooperation.
	This is desired when a single user owns multiple UEs, where then the cooperation among the UEs will benefit the user.
	Even when the UEs belong to different users, cooperation can still be encouraged by incentivizing the UEs or by fulfilling a common task, e.g., charge lower fees to cooperating UEs or operate a common task in military communication \cite{Choi2015}.
	
	To improve the performance of cooperative communication, we utilize the concept of rate-splitting multiple access (RSMA) \cite{RSMA6}.
	RSMA is a multiple access method gaining attention owing to its superior characteristics compared to spatial division multiple access (SDMA) or non-orthogonal multiple access (NOMA) \cite{SDMA1, NOMA1}.
	When serving multiple UEs, SDMA treats the interference as noise while NOMA uses successive interference cancellations (SICs) to fully decode and remove the interference.
	RSMA partially considers the interference as noise and partially decodes the interference, finding the balance in between.
	In specific, the AP divides the messages for the UEs into private and common parts and concatenates the common parts into a common message.
	The AP then encodes the private parts and common message into private and common streams and transmits the streams with linear precoding.
	After reception, the UEs decode the common stream and perform SIC to remove the effect of the common stream. 
	Finally, each UE decodes its private stream with decreased interference.
	By controlling the ratio between the private and common streams, the amount of interference that is decoded or considered as noise can be controlled.
	Thus, RSMA contains SDMA and NOMA as its extreme cases and has superior performance compared to SDMA and NOMA \cite{RSMA1}.
	Due to its superior features, RSMA has been studied through various areas \cite{RSMAMas1,RSMAMas2,RSMAMas3,RSMAMas4,RSMAmm1,RSMAmm2,RSMARIS1}.
	Massive MIMO systems considering imperfect channel state information (CSI), cell-free communication, and mobility were explored in \cite{RSMAMas1,RSMAMas2,RSMAMas3,RSMAMas4}.
	Works considering mmWave systems with RSMA including hybrid beamforming and limited CSI were studied in \cite{RSMAmm1,RSMAmm2}.
	An overview of the early attempts and motivation to consider an RIS and RSMA integrated system was provided in \cite{RSMARIS1}.

	The combination of rate-splitting and cooperative communication called cooperative rate-splitting (CRS) was also explored in \cite{Ours, RSMA2, RSMA3 ,RSMA7}.
	In a cooperative communication system, cooperating UEs assist other UEs by using their resources to relay data to the assisted UEs, while also obtaining their own data.
	As a result, the cooperating UEs decode multiple data streams as they must obtain the data for themselves and for the assisted UEs.
	To efficiently separate the multiple streams, the cooperating UEs require SICs even without rate-splitting. 
	The benefit of combining cooperative communication and rate-splitting then appears, as it does not increase the hardware burden of the cooperating UEs.

	In this paper, we consider a mmWave/THz multi-user (MU) multiple-input single-output (MISO) downlink system.
	Considering the characteristics of high frequency signals, we assume only LoS links for all channels.
	There exists a single UE, denoted as a destination UE (dUE), that has a permanently blocked AP-dUE link.
	The other UEs, denoted as medium UEs (mUEs), are assumed to be mobile.
	Due to the mobility of the mUEs, there exist instances when the AP-mUE or mUE-dUE LoS links are blocked, where we express this feature through a probabilistic LoS model.	
	To successfully serve the dUE, we consider extraction-based CRS (eCRS) in \cite{Ours}, a two-phase cooperative communication framework that utilizes rate-splitting.
	In the first phase, the AP transmits messages to the mUEs.
	The mUEs then extract the message for the dUE from the received messages and transmit the extracted message to the dUE in the second phase.
	This extraction process allows more efficient use of resources compared to the conventional CRS studied in \cite{RSMA2, RSMA3 ,RSMA7}, where the mUEs just relay the common message that also includes information not intended to the dUE.
	For our study, we focus on minimizing the energy consumption with quality-of-service (QoS) constraints.
	The contributions of this paper are listed as follows:
	\begin{itemize}
		\item Due to the probabilistic AP-mUE and mUE-dUE links, we consider transmission over multiple time blocks.
		Through this approach, we can serve all the mUEs and dUE while respecting their delay requirements, whereas communication in a single time block would make serving the blocked mUEs difficult.
		Furthermore, considering multiple time blocks provide an additional benefit of channel diversity. 
		Even with LoS links, the quality of communication depends on the channel gains, e.g., more data can be sent with the same power when the channel gain is larger.
		By considering multiple time blocks, smart resource allocation is possible compared to the single time block case.
		
		\item We propose a new transmission framework named improved distinct eCRS (iDeCRS).
		iDeCRS is a modified version of DeCRS, which is a special case of eCRS in \cite{Ours}.
		By discarding some existing streams and introducing a new common stream, we modify DeCRS so that iDeCRS can take further advantage of rate-splitting with the same hardware constraints compared to the existing DeCRS framework.
		
		\item We formulate and solve an energy consumption minimization problem by assuming genie CSI, denoted as GENIE.
		The genie CSI assumption implies knowing the channels of all the time blocks, which implies the knowledge of the future channels.
		By using the knowledge of the future channels, we derive a lower bound for the energy consumption minimization problem.
		In addition to being used as a performance benchmark, the results of GENIE are also used to assist other proposed techniques.

		\item We propose algorithms to minimize the energy consumption by assuming instantaneous CSI, i.e., knowing only the current channels.
		Different from GENIE, the main challenge of the instantaneous CSI assumption is to allocate the resources efficiently with the lack of future CSI.
		To implement this, we propose two techniques.
		First, we propose efficiency constrained optimization (ECO).
		ECO uses a novel concept of adopting an efficiency constraint in the problem formulation.
		As a result, ECO motivates data transmission until a certain efficiency threshold.
		To the best of our knowledge, this is the first work to address an energy efficiency constrained communication system. 
		Second, we propose even data transmission (EDT).
		EDT is an approach where the data is transmitted as evenly as possible throughout the multiple time blocks.		
		Simulation results show that ECO outperforms EDT when there are many UEs or when the number of time blocks is small, and vise versa.
	\end{itemize}
	
	The rest of paper is organized as follows.
	Section \ref{sec:model} delineates the system model and the proposed iDeCRS framework.
	In Section \ref{sec:GENIE}, we formulate and solve the energy consumption minimization problem with the genie CSI assumption.
	In Section \ref{sec: INST}, we propose and solve ECO and EDT based on the instantaneous CSI assumption.
	Section \ref{sec: analysis} provides initialization techniques and convergence/complexity analyses for the proposed algorithms.
	Section \ref{sec: simul} shows the simulated results of the proposed algorithms.
	Finally, Section \ref{sec:concl} concludes the paper.
	
	\textbf{Notation:} Lower and upper boldface letters represent column vectors and matrices. $\bA^{*}$  and $\bA^{\mathrm{H}}$ denote the conjugate, and conjugate transpose of the matrix $\bA$. 
	${\mathbb{C}}^{m \times n}$ and ${\mathbb{R}}^{m \times n}$ represent the set of all $m \times n$ complex and real matrices. 
	$|{\cdot}|$ denotes the amplitude of the scalar, and $\norm{\cdot}$ represents the $\ell_2$-norm of the vector. 
	$\mathcal{O}$ denotes the Big-O notation. The Kronecker product is denoted by $\otimes$.
	$\boldsymbol{0}_m$ and $\boldsymbol{1}_m$ are used for the $m\times1$ all zero and all one vectors, and $\bI_m$ denotes the $m \times m$ identity matrix.
	$\cC \cN(\bm,\bSigma)$ denotes the circularly symmetric complex Gaussian distribution with mean $\bm$ and variance $\bSigma$.

	\section{System Model}\label{sec:model}
	In this paper, we consider a MU-MISO downlink communication system at mmWave/THz frequencies.
	The system consists of a single AP with $N$ antennas and $(K+1)$ single-antenna UEs.
	Considering the characteristics of the high frequency signals, the channels are assumed to be LoS dominant and vulnerable to blockage.
	Among the $(K+1)$ UEs, we assume a fixed UE, denoted as a dUE, which has a permanently blocked LoS link with the AP.
	The other $K$ UEs, denoted as mUEs, are assumed to be mobile and have AP-mUE and mUE-dUE links that are blocked with a probability\footnote{While the probability of blockage may differ for each link, we assume the same probability $p$ for all links for simplicity.
	Albeit, our proposed techniques delineated afterwards can be applied when the blockage probability differs for each link without any modification.} $p$.
	The $k$-th mUE and dUE have throughput constraints of $\bar{D}_k$ and $\bar{D}_\mathrm{d}$, respectively, which must be satisfied within a delay requirement of $T$ time blocks. In theory, the considered system subsumes the multiple access or relay scenarios by taking $D_\mathrm{d} = 0$ or $D_k = 0, k \in \left\{1, \cdots,K \right\}$, respectively.
		
	Due to the permanently blocked AP-dUE link, the AP cannot communicate directly to the dUE. 
	In order to alleviate this issue, we use DeCRS in \cite{Ours}, which is a cooperative communication framework that works in two phases and exploits rate-splitting.
	The two-phase system model is shown in Fig. \ref{fig:System}.
	In the first phase, the AP transmits private and common streams to the mUEs, where the streams contain the messages for the mUEs and dUE.
	In the second phase, the mUEs decode the common stream, followed by decoding the private stream after performing SIC to reduce interference. 
	Then, the mUEs extract the message for the dUE from the decoded streams and transmit to the dUE. 
	By modifying DeCRS, we propose iDeCRS that fully exploits the potential of rate-splitting with the same hardware conditions.
	
 	\begin{figure}[t] 
		\centering
		\includegraphics[width=1 \columnwidth]{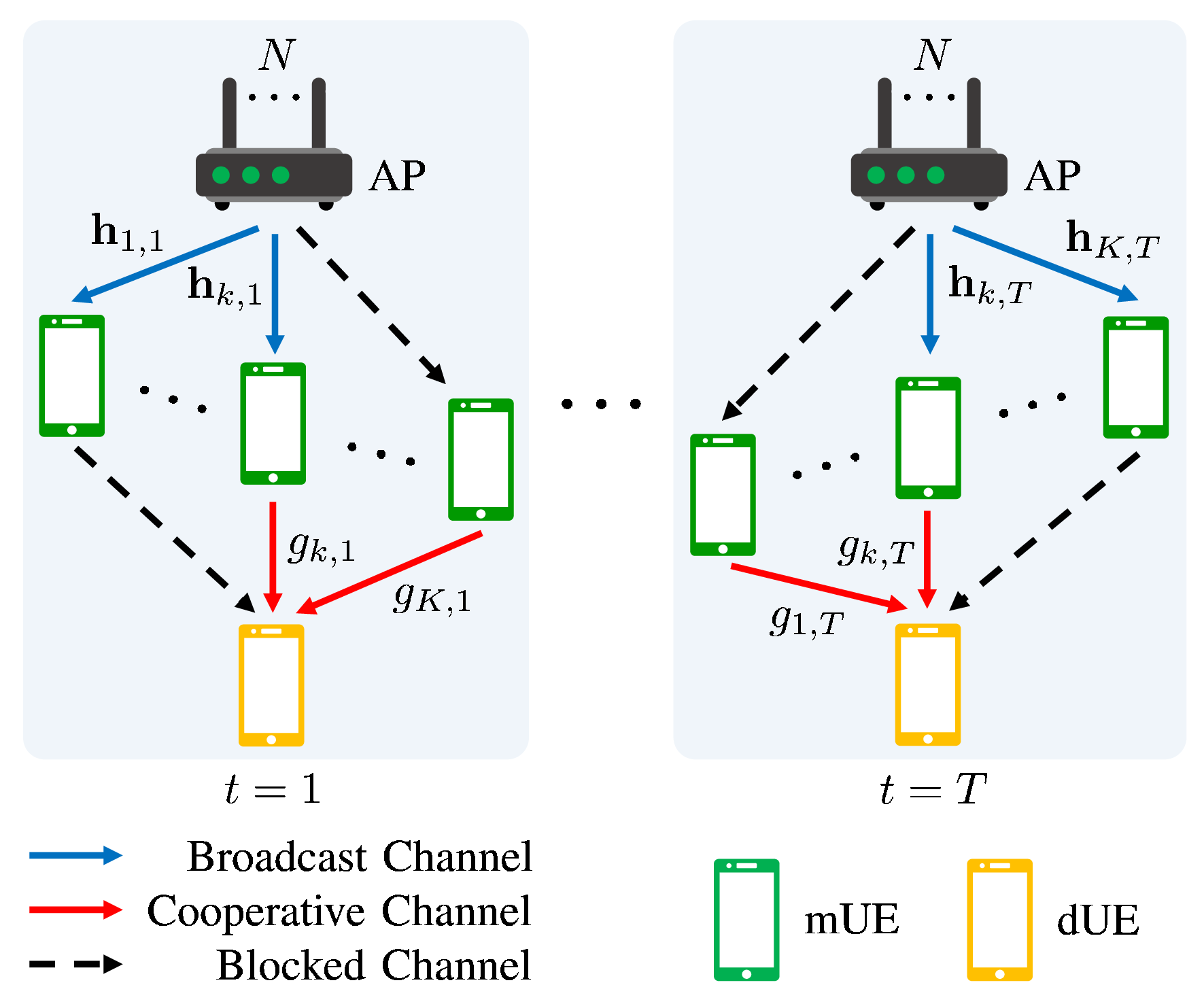}
		\caption{Two-phase cooperative communication system over $T$ time blocks.} 
		\label{fig:System}
	\end{figure}
	
	\subsection{Improved DeCRS (iDeCRS)}
	In this subsection, we focus on the transmission at the $t$-th time block. 
	In the first phase, the AP can only transmit to the mUEs that have AP-mUE LoS links, denoted as $\cK_t = \{1,2,... ,K_t\}$.
	Hence, the AP transmits messages $\left\{W_{k,t}\right\}_{k \in \cK_t}$ and $W_{\mathrm{d},t}$ to the mUEs in the first phase, where $W_{k,t}$ and $W_{\mathrm{d},t}$ are the messages for the $k$-th mUE and dUE, respectively.
	The message $W_{k,t}$ is split into a private part $W_{\mathrm{p},k,t}$ and a common part $W_{\mathrm{c},k,t}$ through a message splitter.
	The dUE message $W_{\mathrm{d},t}$ is first split into $K_t$ messages $\left\{W_{\mathrm{d},k,t}\right\}_{k \in \cK_t}$, where $W_{\mathrm{d},k,t}$ is dedicated for the $k$-th mUE.
	Similar to $W_{k,t}$, $W_{\mathrm{d},k,t}$ is split into a private part $W_{\mathrm{d,p},k,t}$ and a common part $W_{\mathrm{d,c},k,t}$.
	The private parts $W_{\mathrm{p},k,t}$ and $W_{\mathrm{d,p},k,t}$ are combined into a private message via a message combiner and encoded into a private stream $s_{k,t}$ through a message encoder.
	The term \textit{private stream} in our paper implies that the stream is decoded by a single UE.
	The common parts $\{W_{\mathrm{c},k,t}\}_{k \in \cK_t}$ and $\{W_{\mathrm{d,c},k,t}\}_{k \in \cK_t}$ are combined into a common message and encoded into a common stream $s_{\mathrm{c},t}$.
	As a result, the $k$-th private stream $s_{k,t}$ is intended for the $k$-th mUE, and the common stream $s_{\mathrm{c},t}$ is intended for all $K_t$ mUEs.
	
	\begin{figure}%
		\subfloat[The iDeCRS framework at the AP-($k$-th mUE) link in the first phase.]{{\includegraphics[width=0.5\textwidth ]{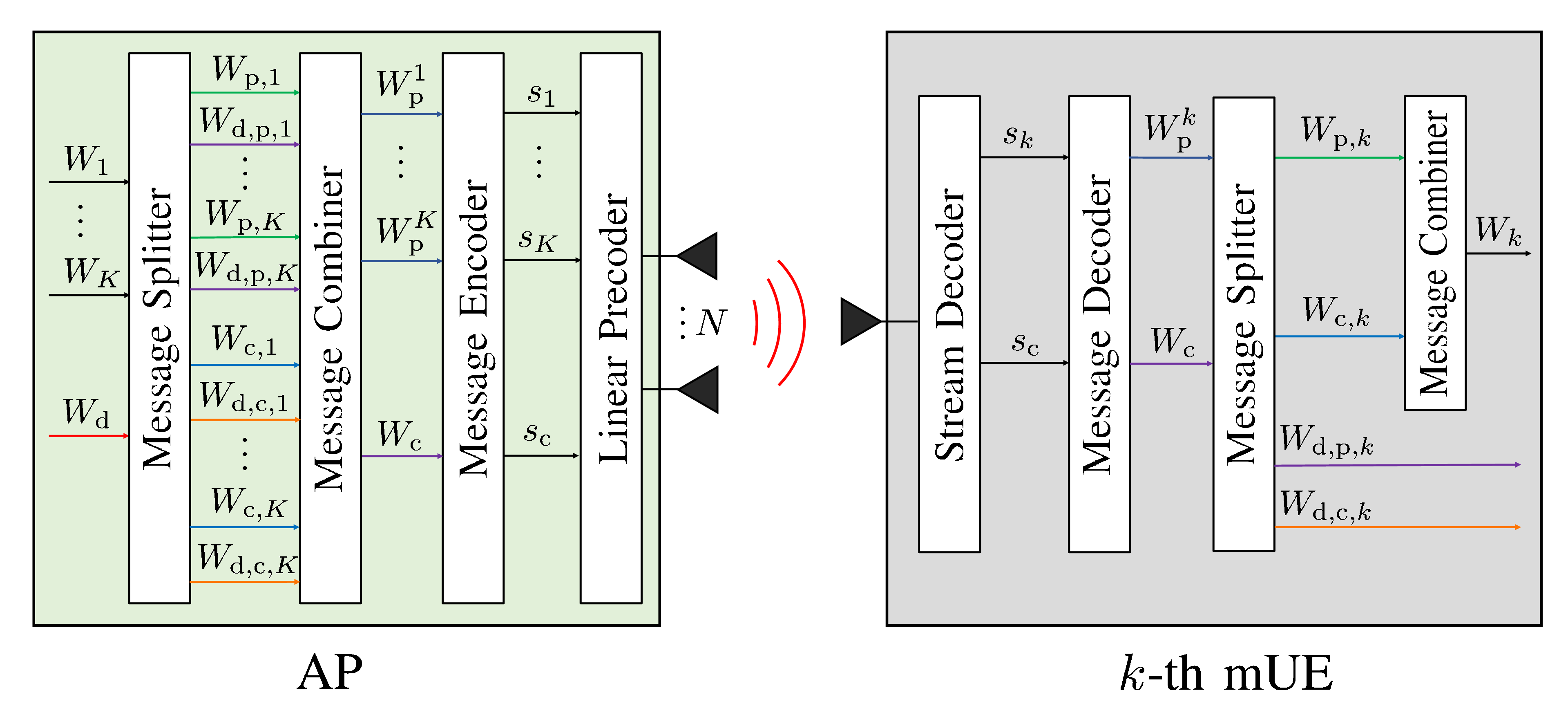}}}%
		\hfil
		\subfloat[The iDeCRS framework at the mUE-dUE link in the second phase.]{{\includegraphics[width=0.5\textwidth ]{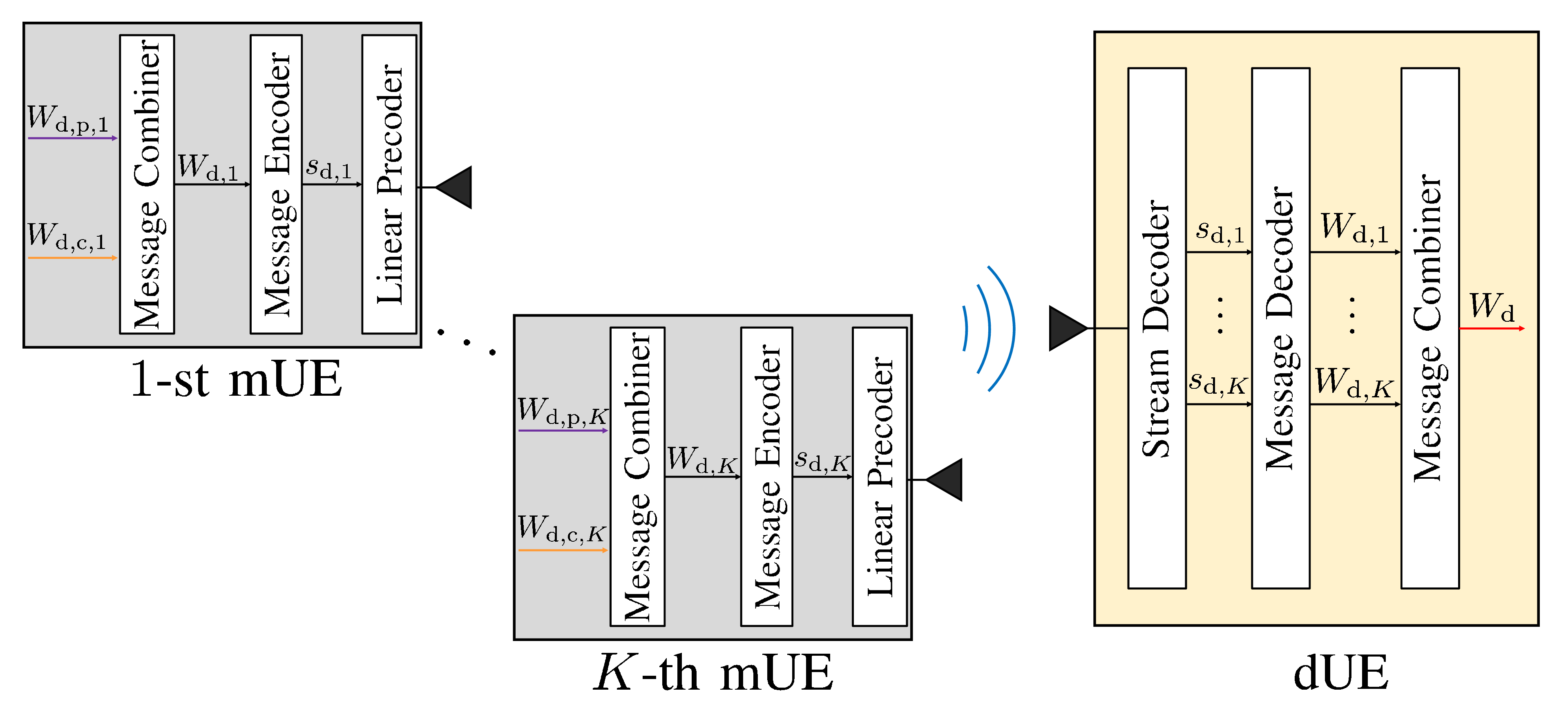} }}%
		\caption{Block diagram of the iDeCRS framework.}%
		\label{fig:iDeCRS}%
	\end{figure}
	
	To provide additional insight, we show that the first phase of the iDeCRS framework boils down to the conventional 1-layer RSMA, where $1$-layer RSMA is a commonly used RSMA framework that requires one SIC and transmits a common message for all UEs \cite{RSMA6}.
	If we consider $W_{k,t}$ and $W_{\mathrm{d},k,t}$ as a single message $W'_{k,t}$, the first phase transmission is effectively identical to the scenario when there is no dUE, i.e., a multiple access scenario without relaying.
	Furthermore, it is clear that using $1$-layer RSMA for the messages $W'_{k,t}$ is identical to the message split of the iDeCRS framework.
	Through this interpretation, we can conclude that the iDeCRS framework enjoys the advantages of $1$-layer RSMA.
	The overall message split of the iDeCRS framework is depicted in Fig. \ref{fig:iDeCRS}.

	\textit{Remark 1:} Compared to the original DeCRS framework, iDeCRS will have a strict improvement in performance.
	The original DeCRS framework utilizes $2K$ streams where there are $K$ common and $K$ private streams.
	For clarity, we denote these streams as \textit{common streams of DeCRS} and \textit{private streams of DeCRS}.
	The $k$-th common stream of DeCRS contains the messages for the $k$-th mUE and dUE, while the $k$-th private stream of DeCRS contains the message for the $k$-th mUE.
	By the definition of DeCRS \cite{Ours}, the $k$-th common and $k$-th private streams of DeCRS are decoded only by the $k$-th mUE, which means that no stream is decoded by multiple mUEs.
	Since each data stream of DeCRS is decoded by only a single mUE, the DeCRS framework does not obtain additional gain from the private streams of DeCRS.
	Effectively, this is identical to a case where a single data stream is split into two separate streams, which does not fully exploit the RSMA features.\footnote{Albeit, DeCRS still has its unique contribution of splitting the dUE message to improve the performance of cooperative communication.}
	Note that, this is not the case for the general eCRS framework, where in general the common streams are decoded by multiple mUEs.
	Therefore, we discard the $K$ private streams of DeCRS and include a single new stream that is decoded by multiple mUEs.
	As a result, the $K$ common streams of DeCRS is identical to the $K$ private streams of iDeCRS and the new common stream of iDeCRS is included, strictly increasing the performance from DeCRS.
		
	After constructing the private and common streams, the AP transmits the streams with linear precoding as
	\begin{align}
		\bx_t = \bff_{\mathrm{c},t} s_{\mathrm{c},t} + \sum_{k=1}^{K_t} \bff_{k,t} s_{k,t},
	\end{align}
	where $\bff_{\mathrm{c},t} \in \mathbb{C}^{N \times 1}$ and $\bff_{k,t} \in \mathbb{C}^{N \times 1}$ are the linear precoders for the common stream and $k$-th private stream, respectively. The $k$-th mUE will then receive the signal given as
	\begin{align}
		y_{k,t} &= \bh_{k,t}^{\mathrm{H}} \bx_{t}+z_{k,t} \notag \\
		&=\bh_{k,t}^{\mathrm{H}} \bF_t \bs_t + z_{k,t},
	\end{align}
	where $\bh_{k,t} \in \mathbb{C}^{N \times 1}$ is the AP-($k$-th mUE) link, $\bF_t = \left[ \bff_{\mathrm{c},t}, \bff_{1,t},...,\bff_{K_t,t}\right]$, $\bs_t = [s_{\mathrm{c}}, s_1,...,s_{K_t}]$, and $z_{k,t} \sim \cC \cN (0,1)$ is the additive white Gaussian noise (AWGN). 
	Without loss of generality, we assume the noise variance is one throughout this paper.
	 
	Similar to current RSMA studies \cite{RSMA1layer1,RSMA1layer2}, the mUEs first decode the common stream.
	The achievable rate of the common stream for the $k$-th mUE is given as
	\begin{align}
	 	R_{\mathrm{c},k,t} = \log_2 \left( 1 +  \frac{|\bh_{k,t}^{\mathrm{H}} \bff_{\mathrm{c},t}|^2}{ \sum_{i=1}^{K_t}|\bh_{k,t}^{\mathrm{H}} \bff_{i,t}|^2 + 1}\right).
	\end{align}
	To guarantee successful decoding for all the mUEs, the common stream must satisfy the constraint given as
	\begin{align}
	 	\sum_{i=1}^{K_t} \{\alpha_{\mathrm{c},i,t} + \beta_{\mathrm{c},i,t} \} \leq R_{\mathrm{c},k,t}, \forall k \in \cK_t,
	\end{align}
	where $\alpha_{\mathrm{c},k,t}$ and $\beta_{\mathrm{c},k,t}$ are the rates of the common parts $W_{\mathrm{c},k,t}$ and $W_{\mathrm{d,c},k,t}$, respectively.
	The messages can be decoded equivalently to the encoding case in reverse order.
	After decoding the common stream, the $k$-th mUE will perform SIC to remove the effect of the common stream  from the received signal.
	The achievable rate of the private stream for the $k$-th mUE is then given as
	\begin{align}
	 R_{\mathrm{p},k,t} = \log_2 \left( 1 +  \frac{|\bh_{k,t}^{\mathrm{H}} \bff_{k,t}|^2}{ \sum_{i\ne k}^{K_t}|\bh_{k,t}^{\mathrm{H}} \bff_{i,t}|^2 + 1}\right),
	\end{align}
	where we observe that the interference induced from the common stream is gone.
	Since the $k$-th private stream is intended only for the $k$-th mUE, the private stream only has to satisfy the constraint given as
	\begin{align}
	 	\alpha_{\mathrm{p},k,t} + \beta_{\mathrm{p},k,t} \leq R_{\mathrm{p},k,t}, 
	\end{align}
	where similar to the common parts, $\alpha_{\mathrm{p},k,t}$ and $\beta_{\mathrm{p},k,t}$ correspond to the rates of the private parts $W_{\mathrm{p},k,t}$ and $W_{\mathrm{d,p},k,t}$, respectively.
	Finally, the rate of the $k$-th mUE can be expressed as $R_{k,t} = \alpha_{\mathrm{c},k,t}+\alpha_{\mathrm{p},k,t}$, and the rate of the dUE message for the $k$-th mUE can be expressed as $R_{\mathrm{d},k,t}^{(1)} = \beta_{\mathrm{c},k,t}+\beta_{\mathrm{p},k,t}$. 
	
	In the second phase, the mUEs extract and relay the messages for the dUE.
	For the $k$-th mUE, the message $W_{\mathrm{d},k,t}$ is obtained through a message combiner.
	Then, the mUE encodes the message $W_{\mathrm{d},k,t}$ into a stream $s_{\mathrm{d},k,t}$ and transmits the signal given as
	\begin{align}
		x_{\mathrm{d},k,t} = f_{\mathrm{d},k,t} s_{\mathrm{d},k,t},
	\end{align}
	where $f_{\mathrm{d},k,t} \in \mathbb{C}$ is the linear precoder for the relaying stream.
	The dUE received signal is then given as
	\begin{align}
		y_{\mathrm{d},t} &= \sum_{k=1}^{K_t} g_{k,t} x_{\mathrm{d},k,t} + z_{\mathrm{d},t} \notag \\
		&= \sum_{k=1}^{K_t} g_{k,t} f_{\mathrm{d},k,t} s_{\mathrm{d},k,t} + z_{\mathrm{d},t},
	\end{align}
	where $g_{k,t}$ is the ($k$-th mUE)-dUE link and $z_{\mathrm{d},t} \sim \cC \cN (0,1)$ is the AWGN.
	We assume that all the streams arrive simultaneously, which is a lower bound for the achievable rate since it will maximize the interference.
	
	The achievable rate for the message from the $k$-th mUE is expressed as
	\begin{align}
		R_{\mathrm{d},k,t}^{(2)} = \log_2 \left( 1+ \frac{|g_{k,t} f_{\mathrm{d},k,t}|^2}{\sum_{i \ne k}^{K_t} |g_{i,t}f_{\mathrm{d},i,t}|^2 +1}\right).
	\end{align}
	Considering the fact that the message $W_{\mathrm{d},k,t}$ should be successfully transmitted through the AP-($k$-th mUE) link and the ($k$-th mUE)-dUE link, the overall rate for the dUE can be expressed as
	\begin{align}
		R_{\mathrm{d},t} = \sum_{k=1}^{K_t} \min \left\{ R_{\mathrm{d},k,t}^{(1)}, R_{\mathrm{d},k,t} ^{(2)} \right\}, \label{eq: dUE rate}
	\end{align}
	where the minimum function is to guarantee successful transmissions in both phases, and \eqref{eq: dUE rate} is the sum of the $K_t$ streams.
	Using the derived rates, the throughput requirements of the $k$-th mUE and dUE can be expressed as
	\begin{align}
		\bar{D}_k &\leq \sum_{t=1}^T \tau B R_{k,t}, \\
		\bar{D}_\mathrm{d} &\leq \sum_{t=1}^T \tau B R_{\mathrm{d},t},
	\end{align}
	respectively, where $\tau$ is the duration of a single time block and $B$ is the bandwidth.
	Considering the power consumption of the AP and the mUEs, the energy consumption at the $t$-th time block is defined as 
	\begin{align}
		P \left(\bF_t, \bff_{\mathrm{d},t} \right) = \tau  \left( \trace \left( \bF_t^{\mathrm{H}} \bF_t \right) + \sum_{k=1}^{K} \trace \left(f_{\mathrm{d},k,t}^{\mathrm{H}} f_{\mathrm{d},k,t}\right)\right).
	\end{align}
	
	To summarize the overall cooperative framework, the AP first transmits $K_t$ private streams and a single common stream that contain the messages for the mUEs and dUE.
	Second, each mUE decodes the common stream and its dedicated private stream, obtaining its message and the dUE message that will be relayed through that mUE.
	Third, each mUE relays the dUE message that it has to the dUE.
	Finally, the dUE obtains its message by decoding the $K_t$ data streams.
	While our main motivation is to serve the dUE via cooperative communication to overcome the blockage and propagation loss of high frequency systems, the proposed iDeCRS framework can still be used in general scenarios where UEs can cooperate to serve a UE in outage.

	\subsection{Channel Model}
	The AP-($k$-th mUE) link at the $t$-th time block is defined as
	\begin{align}
	\bh_{k,t} =
		\begin{cases}
			\sqrt{\chi_0 (d_{k,t}^{\mathrm{AP-mUE}})^{-\eta}} \ba(\phi_{k,t}, \psi_{k,t}), \  &\text{prob.} \ 1-p,\\
			\boldsymbol{0}_N, &\text{prob.} \ p,
		\end{cases}
	\end{align}
	where the link is blocked with the blockage probability $p$.
	The variable $\chi_0$ is the path-loss at the unit distance following the Friis transmission equation, $d_{k,t}^{\mathrm{AP-mUE}}$ is the AP-($k$-th mUE) distance at the $t$-th time block, and $\eta$ is the path-loss exponent.
	The AP is assumed to have a uniform planar array (UPA) structure, and the array response is expressed as
	\begin{align}
		\ba(\phi_{k,t}, &\psi_{k,t}) = [1, \cdots ,\exp \left(j (N_\mathrm{v} - 1 ) \pi \sin \phi_{k,t} \right)]^{\mathrm{T}} \notag \\
		&\otimes [1, \cdots , \exp(j (N_\mathrm{h} -1) \pi \cos \phi_{k,t} \cos \psi_{k,t})]^{\mathrm{T}},
	\end{align}
	where $\phi_{k,t}$ and $\psi_{k,t}$ are the vertical and horizontal angles between the AP and the $k$-th mUE, and the number of vertical and horizontal antennas are given as $N_\mathrm{v}$ and $N_\mathrm{h}$, respectively.
	As the channels are highly directive with minimal scattering due to the high frequency signals, we only consider the LoS links, where the small scale fading effects will be marginal.
	
	Similarly, the ($k$-th mUE)-dUE link at the $t$-th time block is given as
	\begin{align}
		g_{k,t} =
		\begin{cases}
			\sqrt{\chi_0 (d_{k,t}^{\mathrm{mUE-dUE}})^{-\eta}} \exp(j \theta_{k,t}), \  &\text{prob.} \ 1-p,\\
			0, &\text{prob.} \ p,
		\end{cases}
	\end{align}
	where $d_{k,t}^{\mathrm{mUE-dUE}}$ is the ($k$-th mUE)-dUE distance at the $t$-th time block, and $\theta_{k,t}$ is the phase value of the channel.

	\section{Genie CSI Case} \label{sec:GENIE}
	
	 	\begin{figure*}[t] 
		\begin{align}
		\mathrm{(P1.1):} &\min_{\cV'_{\mathrm{GENIE}}} \sum_{t=1}^{T} P \left(\bF_t, \bff_{\mathrm{d},t} \right) \notag \\
		\mathrm{s.t.} \ &\eqref{eq:1-a}\text{-}\eqref{eq:1-c}, \text{(1-e)}, \text{(1-f)}, \notag \\
		&R_{\mathrm{d},t} = \sum_k^{K_t} \mu_{k,t}, \ \mu_{k,t} \leq \beta_{\mathrm{c},k,t} + \beta_{\mathrm{p},k,t}, \   \mu_{k,t} \leq R_{\mathrm{d},k,t}^{(2)}, \tag{1.1-a), (1.1-b), (1.1-c} \label{eq:1.1-a} \\
		&2^{R_{\mathrm{c},k,t}} - 1 \leq \frac{|\bh_{k,t}^{\mathrm{H}}\bff_{\mathrm{c},t}|^2}{\gamma_{\mathrm{c},k,t}}, \ \sum_{i=1}^{K_t}|\bh_{k,t}^{\mathrm{H}} \bff_{i,t}|^2 + 1 \leq \gamma_{\mathrm{c},k,t}, \tag{1.1-d), (1.1-e} \label{eq:1.1-d}\\ 
		&2^{R_{\mathrm{p},k,t}} - 1 \leq \frac{|\bh_{k,t}^{\mathrm{H}}\bff_{k,t}|^2}{\gamma_{\mathrm{p},k,t}}, \ \sum_{i\ne k}^{K_t}|\bh_{k,t}^{\mathrm{H}} \bff_{i,t}|^2 + 1 \leq \gamma_{\mathrm{p},k,t}, \tag{1.1-f), (1.1-g} \label{eq:1.1-f}\\ 
		&2^{R_{\mathrm{d},k,t}^{(2)}} - 1 \leq \frac{|g_{k,t} f_{\mathrm{d},k,t}|^2}{\gamma_{\mathrm{d},k,t}}, \ \sum_{i \ne k}^{K_t} |g_{i,t} f_{\mathrm{d},i,t}|^2 +1 \leq \gamma_{\mathrm{d},k,t}, \tag{1.1-h), (1.1-i} \label{eq:1.1-h}
		\end{align}
		\hrule
	\end{figure*}

	In this section, we formulate and solve the energy consumption minimization problem while assuming the genie CSI case, i.e., the AP and mUEs know the channels of all $T$ time blocks, named GENIE.
	While GENIE is impractical since it assumes the knowledge of the future channels, the main contribution of GENIE is that it provides insight on the potential of transmission over multiple time blocks and acts as a energy consumption lower bound for other techniques.
	The overall problem is formulated as
	\begin{align}
		\mathrm{(P1):} \min_{\cV_{\mathrm{GENIE}}} &\sum_{t=1}^{T} P \left(\bF_t, \bff_{\mathrm{d},t} \right) \notag \\
		\mathrm{s.t.} \ &R_{\mathrm{c},k,t} \geq \sum_{i=1}^{K_t} \left\{\alpha_{\mathrm{c},i,t} + \beta_{\mathrm{c},i,t}\right\}, \tag{1-a} \label{eq:1-a}\\
		&R_{\mathrm{p},k,t} \geq \alpha_{\mathrm{p},k,t} + \beta_{\mathrm{p},k,t}, \tag{1-b} \label{eq:1-b}\\
		&R_{k,t} = \alpha_{\mathrm{c},k,t} + \alpha_{\mathrm{p},k,t}, \tag{1-c} \label{eq:1-c}\\
		&R_{\mathrm{d},t} = \sum_{k=1}^{K_t} \min \{R_{\mathrm{d},k,t}^{(1)}, R_{\mathrm{d},k,t}^{(2)}\},\tag{1-d} \label{eq:1-d}\\
		&D_k \leq \sum_{t=1}^{T} R_{k,t}, \ D_\mathrm{d} \leq \sum_{t=1}^{T} R_{\mathrm{d},t}, \tag{1-e), (1-f} \label{eq:1-e}
	\end{align}
	where we neglect the terms $\forall k \in \cK_t$ and $\forall t \in T$ for simplicity.
	The variables $D_k$ and $D_\mathrm{d}$ are the normalized throughput constraints $D_k = \bar{D}_k/\tau B$ and $D_\mathrm{d} = \bar{D}_\mathrm{d} / \tau B$, respectively. 
	The variable set $\cV_{\mathrm{GENIE}}$ is given as $\cV_{\mathrm{GENIE}} = \{ \bF, \bff_\mathrm{d}, \boldsymbol{\alpha}, \boldsymbol{\beta} \}$ with $\bF = [\bF_1, ..., \bF_T]$, $\bff_\mathrm{d} = [f_{\mathrm{d},1,1}, ..., f_{\mathrm{d},K_T,T}], \boldsymbol{\alpha} = [\alpha_{\mathrm{c},1,1},...,\alpha_{\mathrm{p},K_T,T}],$ and $\boldsymbol{\beta} = [\beta_{\mathrm{c},1,1},...,\beta_{\mathrm{p},K_T,T}]$. 
	The constraints \eqref{eq:1-a} and \eqref{eq:1-b} represent the achievable rates of the common and private streams, respectively. The constraints \eqref{eq:1-c} and \eqref{eq:1-d} express the rates of the mUEs and dUE, respectively. 
	Finally, the constraints (1-e) and (1-f) represent the throughput constraints.
	Owing to the genie CSI, smart resource allocation is possible, e.g., transmit more data when the overall channel gain is larger to decrease the energy consumption of the system.
	Due to the non-convex form of the achievable rates, (P1) cannot be solved directly.
	Instead, we transform (P1) into a tractable form and adopt the successive convex approximation (SCA) approach \cite{SCA}.

	We first transform (P1) into an equivalent problem including slack variables as (P1.1) given at the top of the next page.
	The variable set $\cV'_{\mathrm{GENIE}} = \{\cV_{\mathrm{GENIE}}, \bR,  \boldsymbol{\mu}, \boldsymbol{\gamma} \}$ includes the slack variables $\bR = [R_{\mathrm{c},1,1},...,R_{\mathrm{d},K_T,T}^{(2)}], \boldsymbol{\mu} = [\mu_{1,1},...,\mu_{K_T,T}]$ and $\boldsymbol{\gamma}=[\gamma_{\mathrm{c},1,1},...,\gamma_{\mathrm{d},K_T,T}]$.
	The constraints (1.1-a)-(1.1-c) express the minimum function from the constraint \eqref{eq:1-d}.
	The constraints (1.1-d)-(1.1-i) represent the achievable rates of $R_{\mathrm{c},k,t}, R_{\mathrm{p},k,t}$ and $R_{\mathrm{d},k,t}^{(2)}$, respectively.
	
	While the slack variables do not affect the optimal solution, (P1.1) is still not a standard convex optimization problem.
	Albeit, due to the convex functions on the right-hand-side (RHS) of the constraints (1.1-d), (1.1-f), and (1.1-h), we are motivated to use the SCA approach.
	The resulting surrogate problem is given as
	\begin{align}
		\mathrm{(P1.2):} \min_{\cV'_{\mathrm{GENIE}}} &\sum_{t=1}^{T} P \left(\bF_t, \bff_{\mathrm{d},t} \right) \notag \\
		\mathrm{s.t.} \ &\eqref{eq:1-a}\text{-}\eqref{eq:1-c}, \text{(1-e)}, \text{(1-f)}, \text{(1.1-a)-(1.1-c)}, \notag \\ &\text{(1.1-e)}, \text{(1.1-g)}, \text{(1.1-i)}, \notag \\
		&2^{R_{\mathrm{c},k,t}} - 1 \leq \tilde{g}_{k,t}^{(\ell)} \left( \bff_{\mathrm{c},t}, \gamma_{\mathrm{c},k,t}\right), \tag{1.2-a} \label{eq:1.2-a}\\ 
		&2^{R_{\mathrm{p},k,t}} - 1 \leq \tilde{g}_{k,t}^{(\ell)} \left( \bff_{k,t}, \gamma_{\mathrm{p},k,t}\right), \tag{1.2-b} \label{eq:1.2-b}\\ 
		&2^{R_{\mathrm{d},k,t}^{(2)}} - 1 \leq \tilde{g}_{\mathrm{d},k,t}^{(\ell)} \left( f_{\mathrm{d},k,t}, \gamma_{\mathrm{d},k,t} \right),  \tag{1.2-c} \label{eq:1.2-c}
	\end{align}
	where the first-order Taylor approximation function $\tilde{g}_{k,t}^{(\ell)}$ is defined as
	\begin{align}
		\tilde{g}_{k,t}^{(\ell)} \left( \bff, \gamma \right) = \frac{2\Re\{ \bff^{(\ell)\mathrm{H}} \bh_{k,t} \bh_{k,t}^{\mathrm{H}} \bff \}}{\gamma^{(\ell)}} - \frac{|\bh_{k,t}^{\mathrm{H} } \bff^{(\ell)}|^2}{\left(\gamma^{(\ell)}\right)^2} \gamma,
	\end{align}
	with $\bff^{(\ell)}$ and $\gamma^{(\ell)}$ as the local points at the $\ell$-th iteration.
	The derivation of $\tilde{g}_{k,t}^{(\ell)}$ is provided in Appendix A.
	The function $\tilde{g}_{\mathrm{d},k,t}^{(\ell)}$ is defined similarly, where we neglect the definition due to redundancy. 
	We observe that (P1.2) is now a standard convex optimization problem and can be solved through convex optimization tools such as CVX \cite{CVX1}. 
	By iteratively solving (P1.2) and updating the $\ell$-th local points as the solutions of the $(\ell-1)$-th iteration, the SCA approach guarantees that the solution converges to a local optimum point of the original problem (P1) \cite{SCA}.
	The overall algorithm of GENIE is summarized in Algorithm 1.
	
	\textit{Remark 2:} Since direct communication is impossible between the AP and $k$-th mUE when the AP-($k$-th mUE) link is blocked, we considered communication with $K_t$ mUEs out of $K$ mUEs at the $t$-th time block, where $K_t$ is the number of mUEs that have the AP-mUE links.
	However, we did not explicitly consider the mUE-dUE links, where the $K_t$ mUEs with AP-mUE links may have blocked mUE-dUE links.
	This is because the solution from (P1) innately considers this factor.
	If the ($k$-th mUE)-dUE link at the $t$-th time block is blocked, i.e., $g_{k,t} = 0$, no information for the dUE will be allocated to the $k$-th mUE, i.e., $R_{\mathrm{d},k,t}^{(1)} = 0$. 
	One may argue that similar to the mUE-dUE links, the blockage in the first phase can be analogously considered, i.e., even if we consider communication with all $K$ mUEs at the $t$-th time block, the optimized solution will result in communication with the $K_t$ mUEs by its own.
	However, this is not possible due to the common stream.
	If the problem assumes communication with the mUEs with blocked AP-mUE links, the achievable rate of the common stream is forced to zero, losing all the benefits of rate-splitting.
	
	\textit{Remark 3:} While the considered system assumes a single dUE, it is quite straightforward to consider multiple dUEs.
	By assuming that the time blocks are orthogonal resource blocks and that each dUE in a resource block is a different dUE, (P1) can cover the multiple dUE scenario with a simple modification in the throughput constraints.

	\begin{algorithm}[t]
		\begin{algorithmic} [1]
			\caption{Pseudo code for GENIE}
			\State \textbf{Initialization:} Set $\ell = 0, \bF^{(\ell)}, \bff_\mathrm{d}^{(\ell)},$ and $\bgamma^{(\ell)}$.
			\Repeat
			\State Solve (P1.2) with local points $\bF^{(\ell)}, \bff_\mathrm{d}^{(\ell)},$ and $\bgamma^{(\ell)}$.
			\State Update $\bF^{(\ell+1)}, \bff_\mathrm{d}^{(\ell+1)},$ and $\bgamma^{(\ell+1)}$ as solutions of (P1.2).
			\State Set $\ell = \ell +1$.
			\Until \ Energy consumption decreases by a fraction below a predefined threshold.
		\end{algorithmic}
	\end{algorithm} 
	
	\section{Instantaneous CSI Case} \label{sec: INST}
	In this section, we consider the instantaneous CSI case, i.e., the AP and mUEs know only the current channels.
	We assume perfect CSI for the current channels, where the blockage probability can be obtained by counting the number of instances when the channel gain drops below a predefined threshold, and a practical channel estimation technique adequate to obtain the channel gain has been studied in our previous work \cite{Ours}.
	From hereafter, we propose two techniques that attempt to reduce the energy consumption without the use of future CSI.
	
	\subsection{Efficiency Constrained Optimization (ECO)}
	The essential task for the instantaneous CSI case is to adapt to the various scenarios that may occur at each time block.
	Several examples are given below.
	
	1. The overall channel gains may be small.
	
	2. There may be a small/large number of mUEs with AP-mUE links.
	
	3. The AP-mUE blockage may not be uniformly distributed across the time blocks.
	
	4. Only a small number of mUEs may have good channel qualities.
	
	To adapt to these scenarios, we propose ECO, an energy efficiency constrained optimization algorithm.
	Before we describe ECO, we make an assumption that the system has the knowledge of the blockage probability $p$ and the channel distribution.
	This assumption is easily plausible as the coherence time of high frequency systems is very short and thus, the average channel information is easy to obtain.
	Using this assumption, we can derive the energy efficiency as
	\begin{align}
		\Delta = \frac{D_\mathrm{d} + \sum_{k=1}^{K} D_k}{\sum_{t=1}^{T}P \left(\bF_t, \bff_{\mathrm{d},t} \right)}, \label{eq:Delta}
	\end{align}
	where the numerator and denominator are the total amount of transmitted normalized data and the total amount of energy consumed throughout the $T$ time blocks, respectively.
	The channels and beamformers from \eqref{eq:Delta} are simulated via GENIE.
	Note that, while we use the average energy efficiency from GENIE in \eqref{eq:Delta}, $\Delta$ can be set with the information the AP possesses, e.g., $\Delta$ can be set as the efficiency from the previous transmission.
	
	For ECO, we do not directly minimize the energy consumption. 
	Instead, we attempt to minimize the energy consumption indirectly by constraining the \textit{efficiency} of the system and solving a weighted sum rate maximization problem. 
	By achieving energy efficient communication, this will naturally lead to low energy consumption.
	Using $\Delta$, we can formulate the optimization problem at the $t$-th time block as
	\begin{align}
		\mathrm{(P2):} &\max_{\cV_{\mathrm{ECO}}} \ w_{\mathrm{d},t} R_{\mathrm{d},t} + \sum_{k=1}^{K_t} w_{k,t} R_{k,t} \notag  \\
		\mathrm{s.t.} \ &\eqref{eq:1-a}\text{-}\eqref{eq:1-d}, \notag \\
		&R_{k,t} \leq \delta_{k,t}, \tag{2-a} \label{eq:2-a}\\
		&R_{\mathrm{d},t} \leq \delta_{\mathrm{d},t}, \tag{2-b} \label{eq:2-b}\\
		&\frac{R_{\mathrm{d},t} + \sum_{k=1}^{K_t} R_{k,t}} {P \left(\bF_t, \bff_{\mathrm{d},t} \right)} \geq s \Delta, \tag{2-c} \label{eq:2-c}
	\end{align}
	where $w_{k,t}$ is the weight for the $k$-th mUE given as $w_{k,t} = \delta_{k,t} /  D_k$, and $\delta_{k,t}$ is the residual data at the $t$-th time block as  $\delta_{k,t} = D_k - \sum_{i=1}^{t-1} R_{i,t}$. 
	The weight $w_{\mathrm{d},t}$ is similarly defined.
	Due to the weights in the objective, ECO motivates more data transmissions to the UEs that have more residual data to prevent the AP from transmitting to only the UEs with high channel gains.
	The variable set $\cV_{\mathrm{ECO}}$ is the same as $\cV_{\mathrm{GENIE}}$ but only with the variables at the $t$-th time block.
	The constraints \eqref{eq:2-a} and \eqref{eq:2-b} are the throughput constraints so that the AP does not serve the UEs more than their throughput requirements.
	The constraint \eqref{eq:2-c} is the efficiency constraint, where $s$ is a hyper parameter to relax this constraint.
	As a result, ECO encourages communication until the efficiency becomes $s \Delta$.
	The problem (P2) only considers the $t$-th time block due to the lack of future CSI, whereas (P1) considered all the time blocks simultaneously.
	
	Same as (P1), we can transform (P2) into an equivalent problem as
	\begin{align}
		\mathrm{(P2.1):} \max_{\cV'_{\mathrm{ECO}}} \ &w_{\mathrm{d},t} R_{\mathrm{d},t} + \sum_{k=1}^{K_t} w_{k,t} R_{k,t} \notag  \\
		\mathrm{s.t.} \ &\eqref{eq:1-a} \text{-} \eqref{eq:1-c}, \text{(1.1-a)}\text{-}\text{(1.1-i)}, \eqref{eq:2-a}, \eqref{eq:2-b}, \notag \\
		&R_{\mathrm{d},t} + \sum_{k=1}^{K_t} R_{k,t} \geq s \Delta P\left( \bF_t, \bff_{\mathrm{d},t} \right), \tag{2.1-a} \label{eq:2.1-a}
	\end{align}
	where $\cV'_{\mathrm{ECO}}$ is the same as  $\cV'_{\mathrm{GENIE}}$ but only with the variables at the $t$-th time block.
	The constraint \eqref{eq:2.1-a} is equivalent to \eqref{eq:2-c}, where the denominator is shifted to the RHS.
	Fortunately, the constraints \eqref{eq:2-a}, \eqref{eq:2-b}, and \eqref{eq:2.1-a} are all in standard convex optimization forms.
	Therefore, we solve (P2.1) with the SCA approach similar to (P1.1), which will give a local optimal solution for the problem (P2).
	By solving (P2) for each time block, we obtain a solution for the overall $T$ time blocks.
	Note that, (P2) does not guarantee the throughput requirements of the UEs, e.g., it is possible that there is residual data and all the AP-mUE links are blocked at the last time block.
	This concern will be considered in the next subsection.
	Through the unique approach of energy efficiency constrained communication, ECO can perform smart resource allocation and tackle the various scenarios mentioned before, e.g., transmit more when the channel conditions are good.
	
	The hyper parameter $s$ is a variable to loosen or tighten the efficiency constraint.
	With high $s$ values, the efficiency constraint will be tight.
	However, consider a pessimistic case when many time blocks have detrimental channel conditions.
	Due to the channel conditions, it will be impossible to satisfy the tight efficiency constraint and the throughput requirements simultaneously. 
	Hence, while (P2) may provide an energy efficient solution, it may not suffice the throughput requirements.
	On the contrary, if many time blocks experience good channel conditions, high $s$ values will achieve both the throughput requirements and low energy consumption.
	
	The interpretations for low $s$ values are the opposite to high $s$ values. 
	Low $s$ values will adapt well to detrimental channel conditions.
	However, due to the loose efficiency constraint, the AP will transmit most of the requested data in the early time blocks using more power. 
	Thus, the solution will have low energy efficiency and high energy consumption.
	 
	Exploiting the fact that the throughput requirements will be sufficed in the early time blocks for low $s$ values, a performance bound for low $s$ values can be derived.
	First, we define the throughput transmitted at the $t$-th time block as $\tilde{R}_t = R_{\mathrm{d},t} + \sum_{k=1}^{K} R_{k,t}$.
	Due to the throughput requirements, the equality condition $\sum_{k=1}^{K} D_k + D_{\mathrm{d}} = \sum_{t=1}^{T} \tilde{R}_t$ must be sufficed.
	Then, based on the efficiency constraint \eqref{eq:2.1-a}, the overall energy consumption before the last time block can be expressed as
	\begin{align}
		\frac{\sum_{k=1}^{K} D_k + D_{\mathrm{d}} - \tilde{R}_T}{s\Delta} \geq \sum_{t=1}^{T-1} P\left( \bF_t, \bff_{\mathrm{d},t} \right).
	\end{align}
	As stated before, the throughput requirements will be sufficed before the last time block when $s$ is low, which implies $\tilde{R}_T = 0$ and $P\left( \bF_T, \bff_{\mathrm{d},T} \right) = 0$. 
	Therefore, the bound for the energy consumption is derived as
	\begin{align}
		\frac{\sum_{k=1}^{K} D_k + D_{\mathrm{d}}}{s\Delta} \geq \sum_{t=1}^{T} P \left(\bF_t, \bff_{\mathrm{d},t} \right). \label{eq: ECO bound}
	\end{align}
	We will show that \eqref{eq: ECO bound} indeed works as an energy consumption upper bound and explain possible usage cases of this bound in Section \ref{sec: simul}.
	Note that, the bound in \eqref{eq: ECO bound} may not hold for high $s$ values as the throughput requirements may not be sufficed.
	
	\textit{Remark 4:} While ECO provides satisfactory performance, which is corroborated in Section VI, there is still room for improvement by further exploiting the information of the system.
	For instance, we can set $s$ as a variable to the channel condition for each time block, or modify the efficiency constraint to utilize the information of the channel distribution. 
	These are, however, nontrivial extensions of our current work, and we leave them as future research topics.
	
	\subsection{Even Data Transmission (EDT)}
	In this subsection, we propose EDT, which is an energy consumption minimization algorithm for each time block. 
	The main goal of EDT is to achieve \textit{even} data transmission with respect to time, i.e., split the data for each mUE throughout the time blocks.
	EDT is a naive but persuasive approach, considering the fact that dividing the data evenly over the time blocks is likely to induce efficient communication.
	As a toy example, we consider an energy consumption minimization problem in a single UE AWGN channel.
	
	\begin{proposition}
		For a single UE AWGN channel with the noise distribution as $\cC \cN (0,1)$, even data transmission is the optimal energy consumption minimization strategy.
	\end{proposition}
	\begin{proof}
		The objective is to minimize $\sum_{t=1}^{T} P_t$ with a throughput constraint $D \leq \sum_{t=1}^{T} R_t, \ R_t = \log_2 \left(1+ P_t\right)$. 
		We proceed with the proof through contradiction. 
		Without loss of generality, suppose there exists an optimal solution where $P_1 < P_2$.
		The power $P_1$ can be expressed as $P_1 = f(R_1), f(x) = 2^x - 1$, where $f(x)$ is known to be a convex function.
		By the property of the convex function, $f(x)$ satisfies $f(tx + (1-t)y) \leq t f(x) + (1-t) f(y)$, where $0 \leq t \leq 1$.
		Using this property, we can get $2f(\frac{R_1+R_2}{2}) \leq f(R_1)+f(R_2) = P_1 + P_2$ by setting $x = R_1, y = R_2,$ and $t = 0.5$. 
		Taking $f(\frac{R_1+R_2}{2}) = P^\star$, the energy is decreased from the optimal solution, leading to contradiction.
		Since this holds for all $P_t$, equal power, i.e., even data transmission, is the optimal strategy to minimize the energy consumption.
	\end{proof}

	Motivated by this simple example, we consider EDT.
	The optimization problem of EDT at the $t$-th time block is formulated as
	\begin{align}
		\mathrm{(P3):} \min_{\cV_{\mathrm{EDT}}} &  P \left(\bF_t, \bff_{\mathrm{d},t} \right) \\
		\mathrm{s.t.} \ &\eqref{eq:1-a}\text{-}\eqref{eq:1-d}, \notag \\
		&\frac{\delta_{k,t}}{T-t+1} \leq  R_{k,t}, \tag{3-a} \label{eq:3-a}\\
		&\frac{\delta_{\mathrm{d},t}}{T-t+1} \leq  R_{\mathrm{d},t}, \tag{3-b} \label{eq:3-b}
	\end{align}
	where $\cV_{\mathrm{EDT}} = \cV_{\mathrm{ECO}}$.
	Similar to (P2), (P3) performs optimization for a single time block.
	Since (P3) cannot consider the overall system due to the lack of future CSI, we include alternate throughput constraints as \eqref{eq:3-a} and \eqref{eq:3-b}. 
	For \eqref{eq:3-a}, the rate must be at least over $\delta_{k,t}/\left(T-t+1\right)$, which is the residual data divided by the remaining number of time blocks.
	This constraint is quite intuitive. 
	If the AP-($k$-th mUE) link exists for all time blocks, the constraint induces even data transmission for all time blocks.
	For scenarios when there is blockage, the constraint will encourage more data transmission when large residual data remains. 
	Hence, the constraint induces even data transmission.
	The constraint \eqref{eq:3-b} is analogous to \eqref{eq:3-a}.
	(P3) is equivalent to (P1) for a single time block except the constraints \eqref{eq:3-a} and \eqref{eq:3-b}.
	Thus, (P3) is solved through the SCA approach similar to (P1).

	Due to the random blockage of the channels, both ECO and EDT have a non-zero probability where the throughput requirements cannot be satisfied.
	This is inevitable for the instantaneous CSI assumption and in reality, as the randomness cannot be predicted.
	To compensate for this factor, the AP may try to transmit all the residual data from the $\left(T-\mu\right)$-th time block, i.e., solve EDT with a throughput constraint $\delta_{k,t}$ instead of $\delta_{k,t}/\left(T-\left(T-\mu\right) + 1\right)$, where $\mu$ is a small constant value.
	This approach allows some buffer to mitigate the probability of failure.
	
	\textit{Remark 5:} ECO and EDT have different advantages.
	While ECO adapts to beneficial channel conditions, it does not communicate well if many detrimental channel conditions occur, being problematic to the throughput requirements.
	EDT does not consider the channel conditions while transmitting. 
	Thus, contrast to ECO, it will satisfy the throughput requirements even when the channel conditions are detrimental.
	However, EDT does not have the ability to transmit more data when the channel conditions are beneficial and may suffer from smaller energy efficiency than ECO.


	\section{Initialization and Analyses} \label{sec: analysis}
	\subsection{Initialization}
	To enable the SCA approach, we need to derive an appropriate initial solution inside the feasible set.
	We first delineate the initialization process for GENIE. 
	To satisfy the throughput requirements, we split the data for the $k$-th mUE as $D_k / T_k$, where $T_k$ is the number of time blocks when the AP-($k$-th mUE) link exists.
	This is possible because GENIE has the knowledge of the future channels.
	Similar to the mUEs, we split $D_\mathrm{d}$ into $D_\mathrm{d}/T_{\mathrm{d}}$, where $T_{\mathrm{d}}$ is the number of time blocks with at least one AP-mUE-dUE link.
	For the $t$-th time block, the $K_t$ mUEs and dUE must satisfy the throughput requirements $\{D_k / T_k\}_{k \in \cK_t}$ and $D_\mathrm{d}/T_{\mathrm{d}}$, respectively.
	
	Furthermore, we define the message split according to the iDeCRS framework.
	By splitting the dUE message with the number of mUEs with AP-mUE-dUE links, defined as $K'_t$, the $k$-th mUE at the $t$-th time block receives the rate of $D_k / T_k + D_\mathrm{d} / (T_\mathrm{d} K'_t)$ when the AP-mUE-dUE link exists or $D_k / T_k$ when the mUE-dUE link is blocked.
	With these rate requirements, we initialize the private streams by adopting the energy minimization beamformers assuming SDMA as \cite{Init1}
	\begin{align}
		\bff_{k,t} = \sqrt{p_{k,t}^{\mathrm{init}}} \frac{\left(\bI_N + \sum_{i=1}^{K_t} \lambda_{i,t} \bh_{i,t}\bh_{i,t}^{\mathrm{H}} \right)^{-1} \bh_{k,t}}{ \left\lVert   \left(\bI_N + \sum_{i=1}^{K_t} \lambda_{i,t} \bh_{i,t}\bh_{i,t}^{\mathrm{H}} \right)^{-1} \bh_{k,t} \right\rVert}, \label{eq: init}
	\end{align}
	where $p_{k,t}^{\mathrm{init}}$ and $\lambda_k$ are the private stream power and Lagrange multiplier of the $k$-th mUE at the $t$-th time block, respectively.
	The power values are derived as
	\begin{align}
		\begin{bmatrix}
		p_{1,t}\\
		\vdots \\
		p_{K,t} 
		\end{bmatrix}
		= \bM_t^{-1} \boldsymbol{1}_N,
	\end{align}	
	with the matrix $\bM_t$ defined as
	\begin{align}
		[\bM_t]_{ij} = 
		\begin{cases}
		\frac{1}{\xi_{i,t}} |\bh_{i,t}^{\mathrm{H}} \tilde{\bff}_{i,t}|^2, \quad &i = j, \\
		-|\bh_{i,t}^{\mathrm{H}} \tilde{\bff}_{j,t}|^2, \quad &i \ne j,
		\end{cases}
	\end{align}
	where $i, j \in \cK_t$, $\tilde{\bff}_{k,t}$ is the normalized beamformer from \eqref{eq: init}, and $\xi_{k,t}$ is the minimum signal-to-interference-and-noise ratio (SINR) for the $k$-th mUE at the $t$-th time block.
	The Lagrange multipliers are derived from the fixed-point equations given as
	\begin{align}
		\lambda_{k,t} =\frac{1}{\left( \left( 1 + \frac{1}{\xi_{k,t}} \right) \bh_{k,t}^{\mathrm{H}} \left(  \bI_N + \sum_{i=1}^{K_t} \lambda_{i,t} \bh_{i,t} \bh_{i,t}^{\mathrm{H}} \right)^{-1} \bh_{k,t} \right)}.
	\end{align}
	While the beamformers in \cite{Init1} are optimal considering SDMA, this is not the case for RSMA.
	Next, we define the common stream as \cite{RSMA2}
	\begin{align}
		\bff_{\mathrm{c},t} =  \sqrt{p_{\mathrm{c},t}^{\mathrm{init}}} \bu_{\mathrm{c},t},
	\end{align}
	where $\bu_{\mathrm{c},t}$ is the dominant left singular vector of the channel matrix $\bH_t = \left[ \bh_{1,t}, \cdots, \bh_{K_t,t} \right]$ at the $t$-th time block.
	The power $p_{\mathrm{c},t}^{\mathrm{init}}$ is set as the average power of the private streams.
	For the second phase, we initialize the transmit power of the mUEs as $p_{\mathrm{mUE}} = 0$ dBm.
	
	For EDT, we admit the same initialization with GENIE but with the throughput constraints as \eqref{eq:3-a} and \eqref{eq:3-b}.
	For ECO, to obtain a feasible point, the additional energy efficiency constraint \eqref{eq:2-c} must be satisfied, which is an unprecedented concept.
	By first deriving an initial point as in EDT, we divide the beamformers with a sufficiently large constant $\omega$, i.e., $\bF_t = \bF_t / \omega$ and $\bff_{\mathrm{d},t} = \bff_{\mathrm{d},t} / \omega $, so that the beamformers satisfy the energy efficiency constraint. 	
	
	\subsection{Convergence Analysis}
	We first consider GENIE for the convergence analysis.
	Following the notations of the SCA approach in \cite{SCA}, we define the functions $f(\bF, \bff_\mathrm{d},\bgamma)$ and $g(\bF,\bff_\mathrm{d}, \bgamma | \bF', \bff'_\mathrm{d}, \bgamma')$ as the objective functions of the problems (P1.1) and (P1.2), respectively, with the optimization variables $\bF$, $\bff_\mathrm{d}$, and $\bgamma$.
	For the function $g(\bF,\bff_\mathrm{d}, \bgamma | \bF', \bff'_\mathrm{d}, \bgamma')$, the given terms represent the local points to define the constraints (1.2-a)-(1.2-c) in (P1.2).
	Note that, only the variables related to the SCA approach are denoted for brevity.
	Consider a solution in the $\ell$-th iteration given as $f(\bF^{(\ell)},\bff_\mathrm{d}^{(\ell)}, \bgamma^{(\ell)})$.
	By defining the surrogate problem (P1.2) with the local points as $\bF^{(\ell)}, \bff_\mathrm{d}^{(\ell)},$ and $\bgamma^{(\ell)}$, we get $f(\bF^{(\ell)},\bff_\mathrm{d}^{(\ell)}, \bgamma^{(\ell)}) = g(\bF^{(\ell)}, \bff_\mathrm{d}^{(\ell)}, \bgamma^{(\ell)} | \bF^{(\ell)}, \bff_\mathrm{d}^{(\ell)}, \bgamma^{(\ell)})$, where the equality holds since the objective of both problems are the same.
	In the next iteration, due to the convexity of (P1.2), the optimal solution can be derived as $g(\bF^{(\ell+1)}, \bff_\mathrm{d}^{(\ell+1)}, \bgamma^{(\ell+1)} | \bF^{(\ell)}, \bff_\mathrm{d}^{(\ell)}, \bgamma^{(\ell)})$, where $\bF^{(\ell+1)}, \bff_\mathrm{d}^{(\ell+1)}$ and $\bgamma^{(\ell+1)}$ are the optimal values of the problem (P1.2).
	Since the objective functions are the same, and the constraints of (P1.1) subsume the constraints of (P1.2), it is obvious that $f(\bF^{(\ell+1)}, \bff_\mathrm{d}^{(\ell+1)}, \bgamma^{(\ell+1)}) = g(\bF^{(\ell+1)}, \bff_\mathrm{d}^{(\ell+1)}, \bgamma^{(\ell+1)} | \bF^{(\ell)}, \bff_\mathrm{d}^{(\ell)}, \bgamma^{(\ell)})$. 
	Therefore, we can conclude that the objective is non-increasing as
	\begin{align}
		f(\bF^{(\ell+1)}, &\bff_\mathrm{d}^{(\ell+1)}, \bgamma^{(\ell+1)}) \notag \\
		&= g(\bF^{(\ell+1)}, \bff_\mathrm{d}^{(\ell+1)}, \bgamma^{(\ell+1)} | \bF^{(\ell)}, \bff_\mathrm{d}^{(\ell)}, \bgamma^{(\ell)}) \notag \\
		&\leq g(\bF^{(\ell)},\bff_\mathrm{d}^{(\ell)}, \bgamma^{(\ell)} | \bF^{(\ell)}, \bff_\mathrm{d}^{(\ell)}, \bgamma^{(\ell)}) \notag \\
		&=f(\bF^{(\ell)},\bff_\mathrm{d}^{(\ell)}, \bgamma^{(\ell)}).
	\end{align}
	Since the solutions are non-increasing, and the objective has a lower bound due to the QoS constraints, iteratively solving (P1.2) provides a local optimal solution for the original problem (P1). 
	The same convergence study is applied for ECO and EDT.
	
	\subsection{Complexity Analysis}
	We first analyze the complexity of GENIE. 
	The complexity of (P1.2), solved by convex optimization tools such as CVX or fmincon from MATLAB is quite high due to the exponential constraints \eqref{eq:1.2-a}-\eqref{eq:1.2-c}.
	Albeit, the constraints can be equally expressed as second order cone (SOC) constraints via the SCA approach \cite{SOCP1}. 
	The resulting problem is an SOC programming (SOCP) problem having the complexity of $\cO ([KNT(1-p)]^{3.5})$ \cite{RSMA2}, where $p$ is the probability of channel blockage. 
	With the threshold of convergence as $\epsilon$, the worst case complexity of GENIE is derived as $\cO ([KNT(1-p)]^{3.5} \log(\epsilon^{-1}))$.
	For ECO and EDT, since the optimization is held for each time block, the complexities of ECO and EDT are both $\cO ([KN(1-p)]^{3.5}T \log(\epsilon^{-1}))$.
	We observe that, additional to the practicality of ECO and EDT, both techniques also enjoy the advantage of decreased complexity.
	This advantage is more emphasized when the number of time blocks is large.
	
		\begin{figure}[t] 
		\centering
		\includegraphics[width=1 \columnwidth]{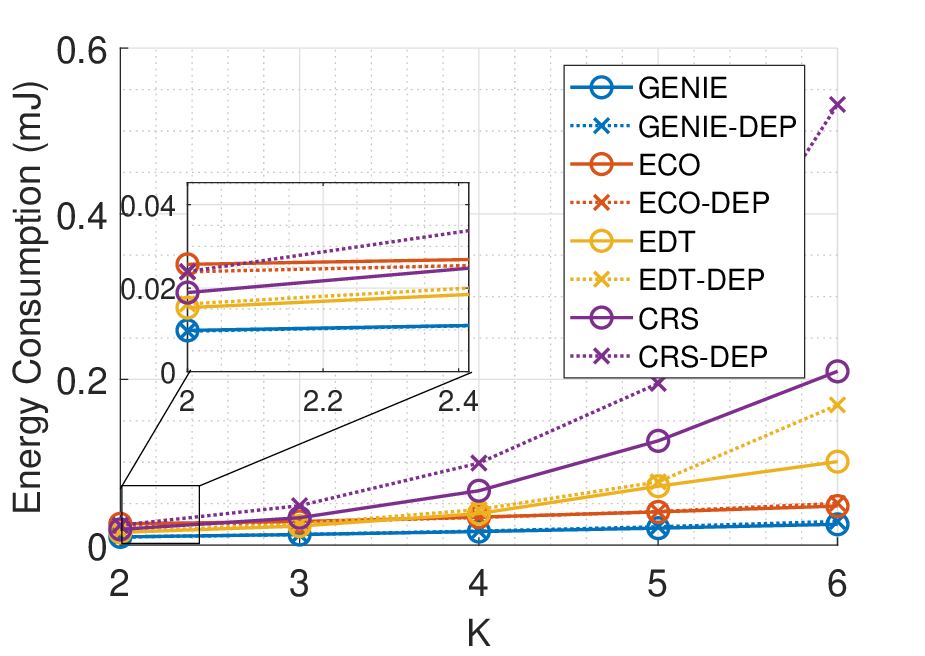}
		\caption{The energy consumption with respect to the number of mUEs, where $T = 20, D_k = 10, D_\mathrm{d} = 2,$ and $s = 0.7$.} 
		\label{fig:AboutK}
	\end{figure}
	\section{Simulation Results} \label{sec: simul}
	In this section, we show the performances of our proposed techniques through simulations. 
	The carrier frequency, duration of a single time block, bandwidth, and path-loss exponent are fixed as $f_\mathrm{c} = 0.3 $ THz, $\tau = 0.4$ ms, $B = 1$ GHz, and $\eta = 2$, respectively.
	The AP and dUE are located at $[0,4,1]$ m and $[8,4,0]$ m, respectively, and the mUEs are assumed to be uniformly distributed for each time block in a box with diagonal coordinates $[2,0,0]$ m and $[6,8,0]$ m.
	The number of AP antennas is $N=16$, and all the channel links have an independent probability of blockage $p = 0.3$ for each time block.
	We simulate three proposed techniques, namely, GENIE, ECO, and EDT.
	In addition, we simulate a benchmark adopting the CRS framework \cite{RSMA2, RSMA3}, where a common stream containing the messages for the mUEs and dUE is relayed to the dUE.
	To observe the effectiveness of the iDeCRS framework and multiple time block consideration, CRS is simulated with the objective function and throughput constraints equal to EDT.
	For fairness, we assume ECO, EDT, and CRS transmit their residual data at the $T$-th time block and that the $T$-th time block has no blockage.

	In Fig. \ref{fig:AboutK}, we plot the energy consumption with respect to the number of mUEs $K$. 
	To observe the effect of correlated blockage, we also simulate the proposed techniques when the blockage probability is dependent to the distance with probability $p = 0.5 d / d_{\max}$, where $d$ is the distance between nodes and $d_{\max}$ is the maximum distance in the simulation environment.
	The distance-dependent blockage cases are denoted as `-DEP.'
	Trivially, the energy consumption increases as the number of mUEs increases for all techniques since this implies an increase of serving UEs.
	We notice that GENIE acts as a lower bound as we expected.
	We observe that EDT is close to GENIE when the number of mUEs is small, but the gap increases as the number of mUEs increases.
	This shows that EDT works well in simple scenarios such as the toy example in Proposition 1.
	The superior performance of ECO compared to EDT, especially when there are many mUEs, implies the benefits of smart resource allocation in multiple time blocks.
	While EDT considers the residual data, it does not consider the channel conditions of each time block.
	In contrast, by considering the quality of the channels of each time block, ECO shows good performance.
	We observe that CRS has higher energy consumption than ECO and GENIE, showing the effectiveness of resource allocation throughout the multiple time blocks.
	Also, the energy consumption of CRS is strictly higher than EDT.
	This is because unlike the eCRS framework, the CRS framework does not extract the dUE message when the mUEs relay the messages, causing the dUE to decode messages unrequired for the dUE.
	Finally, the proposed techniques with distance-dependent blockage follow the same tendency as the independent blockage cases. 
	We also observe that EDT-DEP and CRS-DEP has higher energy consumption than EDT and CRS when $K$ becomes larger.
	For the distance-dependent blockage case, the number of AP-mUE-dUE links is smaller since when the AP-mUE blockage probability is low, the mUE-dUE blockage probability is likely to be high, and vise versa.	
	In consequence, the dUE data is transmitted inefficiently with EDT and CRS compared to ECO or GENIE, signifying the importance of resource allocation.

	\begin{figure}%
		\subfloat[CDF of ECO.]{{\includegraphics[width=0.48\textwidth ]{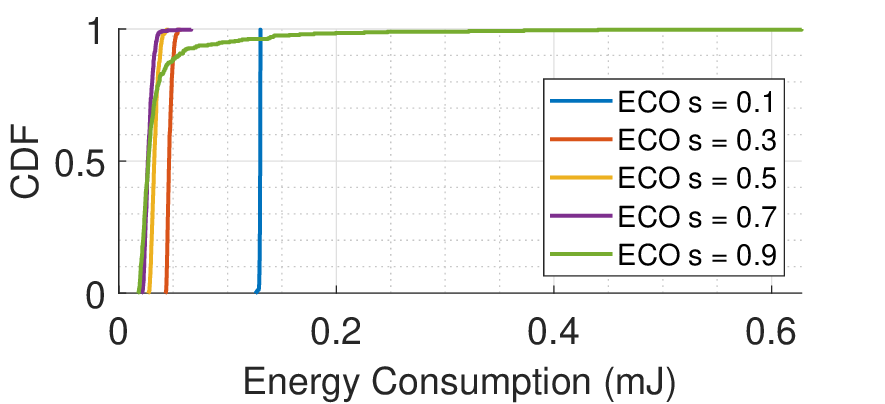} }}%
		\hfil
		\subfloat[Average energy consumption of ECO.]{{\includegraphics[width=0.48\textwidth ]{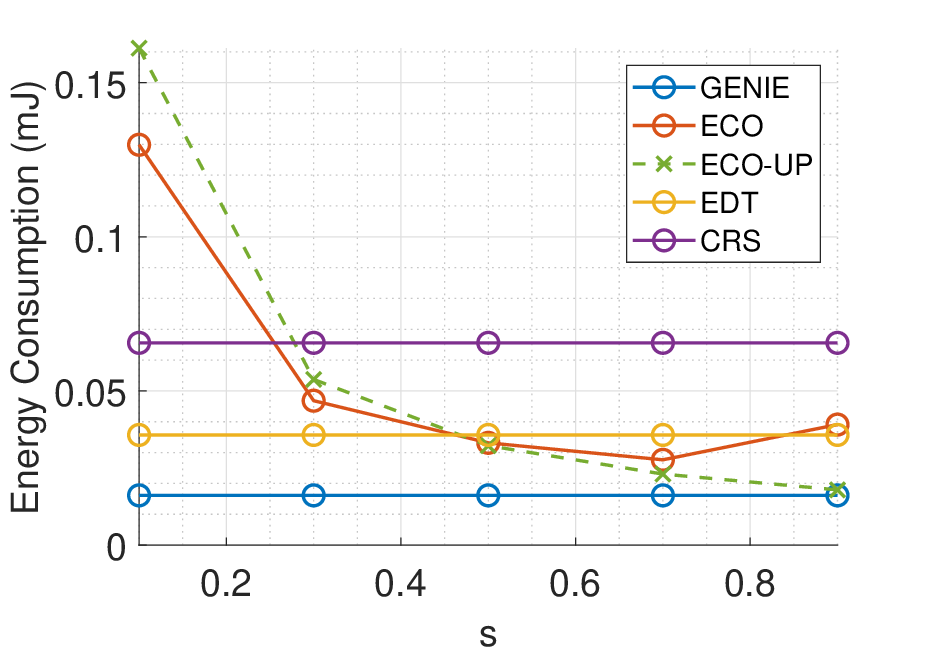} }}%
		\caption{The energy consumption with respect to the hyper parameter $s$, where $T = 20, D_k = 10, D_\mathrm{d} = 2,$ and $K = 4$.}%
		\label{fig:AboutSDCF}%
	\end{figure}

	In Fig. \ref{fig:AboutSDCF} (a), we plot the cumulative density function (CDF) of the energy consumption with respect to the hyper parameter $s$ for ECO. 
	The overall tendency appears to follow our prediction.
	For an extremely low value of $s=0.1$, the performance has no fluctuation due to the loose constraint but the overall performance is low.
	As the efficiency constraint becomes more stringent, the performance increases, but the fluctuation of the instances become larger.
	Due to the stringent constraint, there are more instances where the throughput requirements are not sufficed before the last time block.
	Therefore, a significant amount of energy is used to transmit all the residual data at once at the last time block.
	We observe that a balance between the fluctuation and performance must be found to achieve the minimum \textit{average} energy consumption for a fixed $s$ value. 
	
	In Fig. \ref{fig:AboutSDCF} (b), we plot the average energy consumption with respect to the hyper parameter value $s$. 
	We also plot the energy consumption upper bound for ECO in \eqref{eq: ECO bound} that is effective for low $s$ values, denoted as ECO-UP.
	The average energy consumption of GENIE, EDT, and CRS are constants since they are not affected by $s$.
	As anticipated, the energy consumption is high for low $s$ values, but decreases as $s$ increases.
	The energy consumption then saturates and increases as $s$ becomes too high.
	We observe that the derived bound ECO-UP successfully bounds the performance of ECO for low $s$ values, but starts to deviate as $s$ increases due to the high fluctuation shown in Fig. \ref{fig:AboutSDCF} (a).
	To locate the region of $s$ where the derived bound holds for each instance, a quantitative analysis was conducted to obtain the probability of when the energy consumption is larger than the derived bound.
	Through the results, we confirmed that the derived bound holds until $s \leq 0.3$ with less than a $2$ percent probability of the energy consumption exceeding the derived bound.
	Through the bound and $\Delta$, the hyper parameter $s$ can be chosen considering the usage environment.
	For instance, the hyper parameter $s$ can be chosen slightly high to decrease the average energy consumption if the hardware is robust to energy fluctuation, or can be chosen as a low value otherwise.

	\begin{figure}[t] 
		\centering
		\includegraphics[width=0.95 \columnwidth]{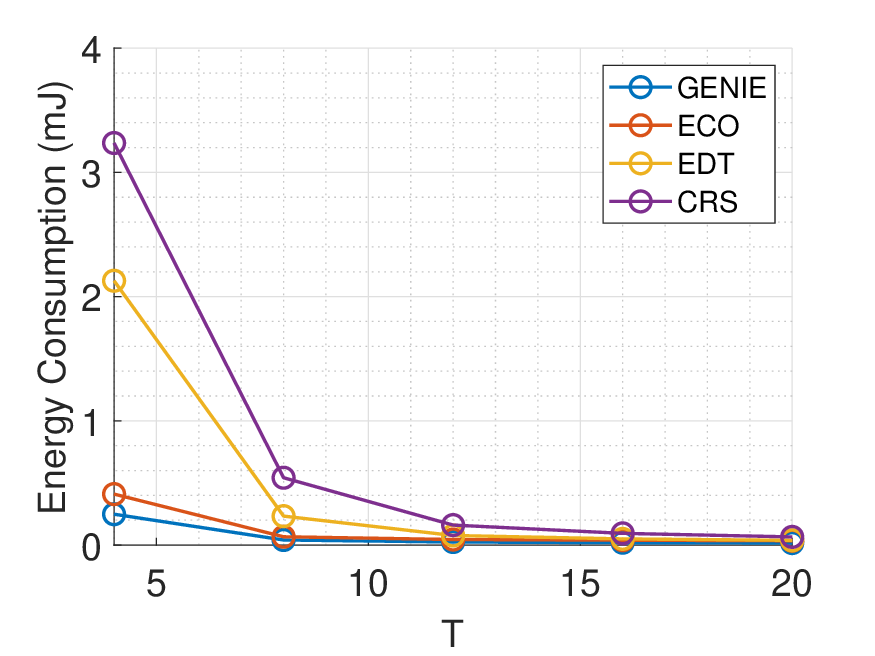}
		\caption{The energy consumption with respect to the number of time blocks, where $s = 0.7, D_k = 10, D_\mathrm{d} = 2,$ and $K = 4$.} 
		\label{fig:AboutT}
	\end{figure}

	In Fig. \ref{fig:AboutT}, we plot the energy consumption with respect to the number of time blocks $T$.
	The energy consumption decreases with the number of time blocks for all techniques, as the data can be spread out.
	Especially, EDT and CRS have a drastic performance increase when the number of time blocks increases.
	This is because with a small number of time blocks, the importance of smart resource allocation is emphasized.
	This shows that ECO is effective for short delay requirement scenarios.
	
	\begin{figure}[t] 
		\centering
		\includegraphics[width=0.95 \columnwidth]{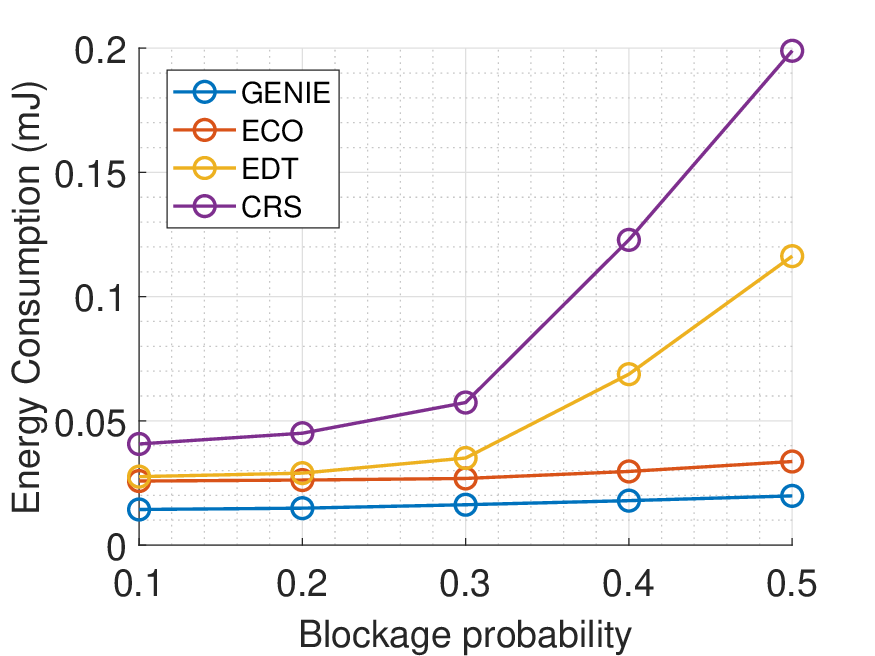}
		\caption{The energy consumption with respect to the blockage probability $p$, where $s = 0.7, D_k = 10, D_\mathrm{d} = 2, T = 20,$ and $K = 4$.} 
		\label{fig:AboutP}
	\end{figure}
	
	In Fig. \ref{fig:AboutP}, we plot the energy consumption with respect to the blockage probability $p$.
	Since the increase in blockage probability means a decrease in the number of usable channels, the energy consumption of all techniques increases as the blockage probability increases.
	Similar to Fig. \ref{fig:AboutT}, we observe that ECO is robust and can handle the increased randomness from the blockage, whereas EDT and CRS do not adapt well to the increased randomness. 

	\begin{figure}[t] 
		\centering
		\includegraphics[width=1 \columnwidth]{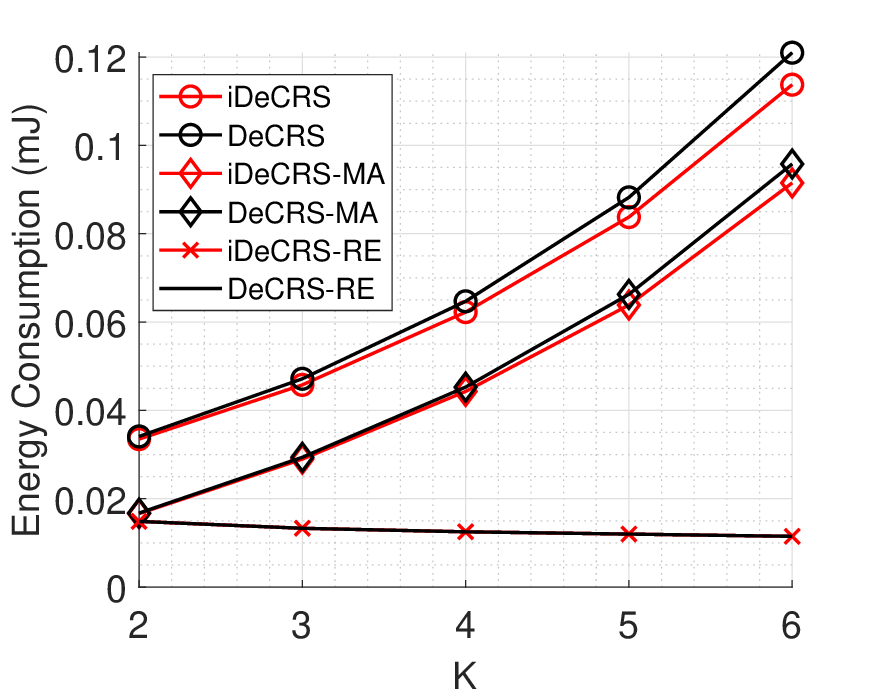}
		\caption{The energy consumption with respect to the number of mUEs, where $T = 20$.} 
		\label{fig:AboutCONV}
	\end{figure}
	
	In Fig. \ref{fig:AboutCONV}, we compare the performances of the iDeCRS framework and the conventional DeCRS framework for different scenarios with genie CSI, where -MA stands for the multiple access scenario without the dUE, and -RE stands for the relaying scenario without the mUE throughput requirements.
	Unless set to zero, $D_k = 20$ and $D_\mathrm{d} = 5$.
	We observe that for all scenarios, the proposed framework outperforms DeCRS, where the performance gap increases as the number of mUEs increases.
	This is because the effect of rate-splitting is more emphasized when more UEs are participating in communication.
	Also, the energy consumption decreases as the number of mUEs increases for iDeCRS-RE and DeCRS-RE since this implies using more relays.
	Finally, the performances of iDeCRS and DeCRS are almost equal in the relaying scenario because of 1) the small required throughput and 2) the fact that the power consumption is mainly from the second phase.
	Considering that both iDeCRS and DeCRS have the order of complexity $\cO ([KNT(1-p)]^{3.5})$ due to the size of the transmit beamformers, iDeCRS outperforms DeCRS with the same hardware conditions and with comparable complexity.
	Overall, we observe that our proposed framework operates well for multiple access and relaying scenarios as well as the cooperative communication scenario.

	\section{Conclusion} \label{sec:concl}
	In this paper, we proposed energy consumption minimization algorithms for high frequency communication systems. 
	To serve permanently blocked UEs, we adopted a CRS approach, where we strengthened DeCRS developed in \cite{Ours}, and proposed the iDeCRS framework.
	To reveal the potential of transmitting over multiple time blocks, we first proposed a energy consumption lower bound which exploits the present and future CSI called GENIE.
	Next, we proposed ECO, a novel efficiency constrained algorithm, followed by EDT, a naive but effective even data transmission approach.
	Through the simulation results, we observed that ECO has good performance when there are many participating mUEs and EDT has considerable performance when many time blocks are available, while GENIE works as a legitimate lower bound.
	The results also show that the proposed techniques are suitable for both multiple access and relay scenarios.
	From the results, we conclude that judiciously exploiting multiple time blocks in high frequencies is promising for energy efficient communication.
	While we took some heuristic approaches to tackle the concept of multiple time blocks in this paper, tailoring other theoretical approaches, e.g., Markov decision process \cite{MDP1}, to our model would be an interesting future work.

	\section*{Appendix A}
	For the function $f(\bx,y) = \frac{\bx^{\mathrm{H}} \bh \bh^{\mathrm{H}} \bx}{y}$, the first-order Taylor expansion is defined as
	\begin{align}
	\tilde{f}(\bx,y) &= f(\bx_0, y_0) + \frac{\partial f(\bx_0, y_0)}{\partial \bx} (\bx - \bx_0) \notag \\+ &\frac{\partial f(\bx_0, y_0)}{\partial \bx^*} (\bx^* - \bx_0^*)+\frac{\partial f(\bx_0, y_0)}{\partial y} (y - y_0),
	\end{align}
	with the local points $\bx_0, y_0$ and the first-order derivatives given as \cite{Diff1}
	\begin{align}
	&\frac{\partial f(\bx_0, y_0)}{\partial \bx} = \frac{\bx_0 \bh \bh^{\mathrm{H}}}{y_0},\\
	&\frac{\partial f(\bx_0, y_0)}{\partial \bx^*} = \frac{\bx_0^\mathrm{T} \bh^* \bh^{\mathrm{T}}}{y_0},\\
	&\frac{\partial f(\bx_0, y_0)}{\partial y} = - \frac{\bx_0^{\mathrm{H}} \bh \bh^{\mathrm{H}} \bx_0}{y_0^2}.
	\end{align}
	By substituting the derivatives, the first-order Taylor expansion is derived as
	\begin{align}
	\tilde{f}(\bx,y) = 2 \mathrm{Re} \left\{\frac{\bx_0^{\mathrm{H}} \bh \bh^{\mathrm{H}} \bx}{y_0}\right\} - \frac{\bx_0^{\mathrm{H}} \bh \bh^{\mathrm{H}} \bx_0}{y_0^2}y,
	\end{align}
	finishing the proof.

	\bibliographystyle{IEEEtran}
	\bibliography{Ref}

\end{document}